\documentclass{article} 

\makeatletter

\PassOptionsToPackage{hyperref,svgnames,table}{xcolor}
\PassOptionsToPackage{amsmath,hyperref,standard,thmmarks}{ntheorem}
\PassOptionsToPackage{allcolors=Blue,colorlinks}{hyperref}
\PassOptionsToPackage{T1}{fontenc}

\usepackage{algorithm}
\usepackage{amsmath}
\usepackage{amssymb}

\usepackage{fontenc}
\usepackage{graphicx}
\usepackage{hyperref}
\usepackage{ntheorem}
\usepackage{paralist}
\usepackage{subcaption}
\usepackage{tabularx}
\usepackage{xcolor}

\usepackage{cleveref}

\newcommand\mkYcAoM[1][]{\mathcal{D}_{#1}} 
\newcommand\mpBykPY[1][]{\mathcal{U}_{#1}} 
\newcommand\mYsfxOR[1][X]{\mathcal{#1}}
\newcommand\mNjOrhQ[1][]{#1^{\star}}
\newcommand\mmHxeoG[1]{\operatorname{\mathcal{O}}(#1)}
\newcommand\majimMG[1]{|#1|}
\newcommand\mfBPXHb[1]{#1^\intercal} 
\newcommand\mPlzHde[1]{#1^+}
\def\norm#1{\lVert#1\rVert}
\newcommand\mGctqYz[1][X]{\mathcal{C}(#1)}
\newcommand\mCBuiuW[1]{2^{#1}} 
\newcommand\mRxFzBg[1][X]{\mathcal{T}(#1)}
\newcommand\mcrHMQW[1]{[\![#1]\!]}
\newcommand\mAoguFS[1]{\mathbf{#1}}

\newcommand\mijCQyI[1][]{V^{#1}}
\renewcommand\vec[1]{\mathbf{#1}}
\newcommand\mDpbDSN[1][]{W^{#1}}

\newcommand\mPmwBxY[1][]{J^{#1}}
\newcommand\mBuHESu[1][\delta]{\Delta_{#1}}
\newcommand\mNcTCKG[1][f]{\mathit{#1}}
\newcommand\mTorQlx[1]{\mathsf{#1}}

\newcommand\mORCysO[1][]{\mathcal{O}_{\mathsf{#1}}}
\newcommand\mQtakuY[1][x]{\bar{#1}}

\makeatletter
\begin{document}
\title{Solvability of Approximate Reach-Avoid Games}
\author{Mario Gleirscher}
\maketitle
\label{l:1}

\begin{abstract}
  \emph{Objective:} In a companion paper, we propose a parametric
  hybrid automaton model and an algorithm for the online synthesis of
  robustly correct and near-optimal controllers for
  cyber-physical system 
  with reach-avoid guarantees. 
  A key part of this synthesis problem is based on a weighted
  discretised 
  game 
  and solved via scope-adaptive discrete dynamic programming.
  \emph{Approach:} This work examines proofs of key properties of the
  discussed algorithm.
  \emph{Evaluation:} The main proof is by induction over the stages of
  a discrete Hamilton-Jacobi-Bellman system of equations.
  \emph{Contribution:} The results include a game solvability theorem
  and identify necessary and sufficient conditions for its
  applicability.
\end{abstract}

\section{Introduction}

\paragraph{Motivation.}

Control of non-linear, disturbed dynamical systems can be solved by
constructing provably robust~(i.e., for safety) and
near-optimal~(i.e., for minimum-cost reachability)
controllers.  To accommodate parametric uncertainties, such
controllers have to be synthesised online, that is, during operation.

In a companion paper~\cite{Gleirscher2025-ParametricModelOptimal},
we formalise such a control problem as a parametric weighted hybrid game automaton~(HGA)
with reach-avoid winning conditions.  The automaton enables a flexible
three-tier separation of control concerns: supervisory (e.g., moving
obstacle evasion), tactical (e.g., near-optimal trajectory tracking),
and sub-tactical 
(e.g., local stability) control.
We provide a bounded solution algorithm for the discretised modal
games of the HGA.  The algorithm uses step-shielded discrete dynamic programming~(DDP) and
supports non- and quasi-stationary policies, increasing robustness to
disturbance 
at the cost of performance 
\cite{Lewis2012-ReinforcementLearningFeedback}.  
To prepare our algorithm for online use, we (slightly) reduce DDP's
curse-of-dimensionality problem by scope adaptation\footnote{The
  relevant fraction of the state space is selected according to the
  current system state and extended on-the-fly to increase game
  solvability.}  and fixpoint approximation for the winning region.
The modal synthesis algorithm is part of a HGA player.

\paragraph{Challenge and Research Question.}

Given a system described as a HGA, part of the system's
assurances can be based on the guarantees provided by the solutions of
the parametric modal games.  Based on our exhaustive numerical
algorithm for near-optimal online synthesis, under which assumptions
does the algorithm construct a correct solution providing a robust
reach-avoid guarantee?

\paragraph{Approach and Application.}

Our notion of approximation encompasses
a bound on the number of stages, 
a control robustness margin, 
as well as state space discretisation 
and dynamic scoping. 
We first examine necessary and sufficient conditions for game
solvability and then use these conditions to provide a solvability
proof.  The HGA, the algorithms (their {C\texttt{++} implementation), and
an experimental evaluation are discussed in more detail in
\cite{Gleirscher2025-ParametricModelOptimal}.

\paragraph{Contributions.}

In addition to our model and experiments in
\cite{Gleirscher2025-ParametricModelOptimal}, this work
\begin{itemize}
\item provides a \emph{solvability} (i.e., controller existence)
  theorem using a numerical Hamilton-Jacobi-Isaacs~(HJI) equation
  \cite{Tomlin2000-gametheoreticapproach} as a generic inductive
  invariant of the incrementally computed winning
  region 
  and
\item 
  identifies several conditions under which the synthesised 
  controllers guarantee $\delta$-robust bounded correctness, in
  particular, safety for at least $N$ stages and liveness by
  reaching goals within $N$ stages; both under rejection of
  delayed bounded disturbance.
\end{itemize}
Under the identified assumptions, we obtain a decision
procedure 
for approximate modal reach-avoid games using a DDP-based
$\mmHxeoG{N\majimMG{{\mathbb{X}}}\majimMG{\mpBykPY}\majimMG{\mkYcAoM}}$-algorithm.

\paragraph{Outline.}

\Cref{l:2} elaborates on the framework used in
\cite{Gleirscher2025-ParametricModelOptimal}.  In
\Cref{l:13}, we discuss the solvability
conditions for our reach-avoid synthesis with guarantees.  Based on
these conditions, \Cref{l:20} presents a solvability
theorem and a control barrier certificate.  The
\Cref{l:27,l:28} reflect on the approach
and summarise key findings.  Details on the proof are presented in
\Cref{l:29}.

\section{Preliminaries}
\label{l:2}

This section introduces terms and definitions and finishes with a
specification of the synthesis problem building on the solvability of
reach-avoid games.

\paragraph{Notation.} 

For a set of variables
$\mathit{Var}=X\uplus U\uplus D$,
let ${\mathbb{X}}\subset\ensuremath{\mathbb{Z}}^{X}$ be an $X$-typed,
finite,
$m$-dimensional Euclidean \emph{state space},\footnote{For a 3D
  coordinate $\vec p$, we use the typical naming convention
  $\vec p=\mfBPXHb{(x,y,z)}$.}  and let $\mpBykPY$ and $\mkYcAoM$
be $U$- and $D$-typed \emph{control} and
\emph{disturbance} ranges, each including $\vec 0$.
Moreover, let 
the classes of terms $\mRxFzBg[\mathit{Var}]$ 
and constraints
$\mGctqYz[\mathit{Var}]$ over $\mathit{Var}$.
$\norm\cdot$ is the corresponding 2-norm, $\langle\cdot\rangle$ the
scalar product, 
$A\oplus B=\{a+b\mid a\in A\land b\in B\}$ the Minkowski sum of
$A,B\subseteq{\mathbb{X}}$, and $A|_{\mathit{Var}'}$ the projection of
tuples 
in $A$ to variables in $\mathit{Var}'\subseteq\mathit{Var}$,\footnote{We
  omit set parentheses when referring to singleton sets 
  in subscripts.}  and $\dashrightarrow$ indicates a partial map.
Point-wise operators 
are lifted as usual.\footnote{For example, $\min,\max$ on vectors of
  sets are evaluated element-wise and return a vector of scalars.
  Negation of a set returns the set of negated elements.}  Moreover,
$\underline I$ and $\bar I$ denote the lower and
upper bounds of an interval $I\subseteq\ensuremath{\mathbb{R}}$ and
$\underline I..\bar I$ signifies that
$I\subseteq\ensuremath{\mathbb{Z}}$.
$\mcrHMQW\varphi_{\mathfrak{M}}$ denotes $\varphi$'s models in
domain~${\mathfrak{M}}$,\footnote{For example, the
  class of states, state pairs, sequences, or transition systems.}
formally,
$\mcrHMQW\varphi_{\mathfrak{M}} = \{M\in{\mathfrak{M}}\mid M\models\varphi\}$.

\paragraph{Trajectories.} 
\label{l:3}

Given the finite \emph{sequences} $\mNjOrhQ[{\mathbb{X}}]$ over
${\mathbb{X}}$, the concatenation $\mQtakuY\mQtakuY'$ of
$\mQtakuY,\mQtakuY'\in\mNjOrhQ[{\mathbb{X}}]$, the length $\majimMG\mQtakuY$ of
$\mQtakuY$, and its value $\mQtakuY_k$ at position
$k\in1..\majimMG\mQtakuY$,
we call $\mQtakuY$ a \emph{$\mBuHESu[]$-trajectory} if
$\forall k\in1..\majimMG\mQtakuY-1\colon \mQtakuY_{k+1}\in\mQtakuY_k\oplus
\mBuHESu[]$ with convex
$\mBuHESu[]\subseteq{\mathbb{X}}$ constraining changes per
transition.\footnote{For example,
  $\mBuHESu[]=\{-2,1,0,1,2\}^m=[-2,2]^m=[\pm 2]^m$ restricts changes to
  the discrete centered hyper-cube of length~5.}  For convenience,
$\mQtakuY$ is also used as the set of its states.  $\mQtakuY$ is said to
be \emph{${\mathbb{X}}$-continuous} 
if $\mBuHESu[]$ enforces the minimal observable change in all changed
${\mathbb{X}}$-components (e.g., $\mBuHESu[]=[\pm1]^m$ for
$\ensuremath{\mathbb{Z}}$).\footnote{It can be shown that trajectory
  ${\mathbb{X}}$-continuity is preserved under
  projection.} 
Let $\bar{\mathbb{X}}(\mBuHESu[]),\tilde{\mathbb{X}}\subset\mNjOrhQ[{\mathbb{X}}]$
be the classes of $\mBuHESu[]$-trajectories and ${\mathbb{X}}$-continuous
trajectories $\tilde{\mathbb{X}}=\bar{\mathbb{X}}([\pm 1]^m)$, respectively.
We write~$\bar{\mathbb{X}}$ if $\mBuHESu[]$ is clear from the context.
Given $\mQtakuY\in\bar{\mathbb{X}}(\mBuHESu[])$,
$\mathop\mathit{ct}(\mQtakuY) =
\arg\min_{\mQtakuY'\in\tilde{\mathbb{X}},\,\mQtakuY\subseteq\mQtakuY'}\norm{\mQtakuY'}$
provides a shortest continuous trajectory $\mQtakuY'\supseteq\mQtakuY$.
We obtain $\mQtakuY'\in\bar{\mathbb{X}}(\mBuHESu[])$ by the triangle
inequality. 

\paragraph{Weighted Hybrid Game Automata.}

Given a set of modes $Q$,
an 
alphabet~$A$, events
$E\subseteq
Q\times A\times Q$, and the hybrid
state space ${\mathcal{S}} = Q\times{\mathbb{X}}$,
a weighted HGA
$G =
(\mathit{Gra},\mathit{Var},\mathit{Ini},\mathit{Inv},\mNcTCKG,\mathit{Jmp},F)$
comprises a
\begin{inparaitem}[]
\item mode transition graph
  $\mathit{Gra} = (Q, A,$ $E)$,
\item initial conditions $\mathit{Ini}\colon Q\to\mGctqYz$,
\item invariants
  $\mathit{Inv}\colon Q\to\mathcal{C}(X)$,
\item flow constraints
  $\mNcTCKG\colon
  {\mathcal{S}}\to\mGctqYz[\mathit{Var}\cup\dot{X}]$,
\item jump conditions
  $\mathit{Jmp}\colon E\to\mGctqYz[\mathit{Var}\cup
  \mPlzHde{X}]$
  comprised of guards and updates, and
\item weighted 
  winning conditions
  $F\colon{\mathcal{S}}\to\mCBuiuW{\mRxFzBg[\mathit{Var}]}$,
  generalising
  final conditions $\mathit{Fin}\colon Q\to\mGctqYz$
  \cite{Henzinger2000-theoryhybridautomata}.
\end{inparaitem}
The copies $\dot{X}$ and $\mPlzHde{X}$ refer to the time
derivatives and discrete updates of $X$.  $\operatorname{grd}(e)$ and
$\operatorname{upd}(e)$ denote the guard and update conditions of
$\mathit{Jmp}(e)$.  In the remainder, we frequently use
${\vec{s}}=(q,{\vec{x}})\in{\mathcal{S}}$ for a state
of~$G$.

\paragraph{Integer Difference Games.}

Let $\Pi_{\mYsfxOR[Y]}={\mathbb{X}} 
\times\mathbb{N}_+\dashrightarrow\mYsfxOR[Y]$ be a \emph{strategy space} for a
player with alphabet~$\mYsfxOR[Y]$.
Given a \emph{scope} $\mathcal{X}\subseteq{\mathbb{X}}$, a flow
constraint~$\mNcTCKG$,
and its discretised Taylor
expansion~$\hat\mNcTCKG$,
a state ${\vec{s}}$ of~$G$ can be associated with a
discretised
two-player game
$G_{\vec{s}}=(\mathcal{X},
\hat\mNcTCKG,F)$.\footnote{Players share the
  state and time but do not know each others' future
  actions.}  We use ${\vec{u}}\colon\Pi_{\mpBykPY}$ and
${\vec{d}}\colon\Pi_{\mkYcAoM}$ to describe
the 
memory-less 
strategy profile of the two players,
controller and environment, based on the
non-empty 
control and disturbance ranges $\mpBykPY$ and $\mkYcAoM$.

For the weighted \emph{reach-avoid} winning condition
$F=(L,\Phi)$ with
stage 
and terminal costs~$L$ and~$\Phi$, we employ
\begin{align}
  \label{l:4}
  L({\vec{u}},{\vec{d}};{\vec{s}},k) =
  \begin{cases}
    0,\;\text{if}\;\rho\land\neg\alpha
    \\
    \top,\;\text{if $\alpha$ (shielded)}
    \\
    \lambda({\vec{u}},{\vec{d}};{\vec{x}},k)
    ,\;\text{otherwise}
  \end{cases}
  \Phi({\vec{s}}) =
  \begin{cases}
    0,\;\text{if}\;\rho\land\neg\alpha\\
    \top,\; \text{otherwise}
  \end{cases}
\end{align}
where $\rho,\alpha\colon\mGctqYz[\mathit{Var}]$ specify the goal
and unsafe regions to be reached and avoided, respectively.  We assume
that $\mcrHMQW\rho\subseteq\mathcal{X}$ and
$\mcrHMQW\alpha\subseteq{\mathbb{X}}$ with
$\mcrHMQW\rho\setminus\mcrHMQW\alpha\neq\varnothing$.
This way, $L$ penalises $\alpha$ and rewards $\rho$,
maximally.  Furthermore,
$\lambda({\vec{u}},{\vec{d}};{\vec{x}},k) = \mfBPXHb{\vec{x}}
P{\vec{x}} + \mfBPXHb{\vec{u}} Q{\vec{u}} + \mfBPXHb{\vec{d}} R{\vec{d}}$ is a
weight term with correspondingly dimensioned matrices $P$, $Q$,
and~$R$.\footnote{As we do not restrict ourselves to 
  LQR 
  solutions, $\lambda$ can be non-convex.} 
Let $\mathcal{G}$ be the class of integer difference games.

A \emph{winning strategy} for $G_{\vec{s}}$ is a
finite-horizon, 
discrete optimal 
controller ${\vec{u}}^*\colon\Pi_{\mpBykPY}$, which can be
obtained by solving a 
constrained, discrete dynamic optimisation
problem~\cite{Bertsekas2017-DynamicProgrammingOptimal-v1} for given a
horizon $N$:
\begin{align}
  \label{l:5}
  &&{\vec{u}}^*(\bar{\vec{x}};k)
  &= \arg\min_{{\vec{u}}\in\Pi_{\mpBykPY}}\max_{{\vec{d}}\in\Pi_{\mkYcAoM}}
    \mPmwBxY(\bar{\vec{x}},{\vec{u}},{\vec{d}},k)
  \\
  \text{and}\qquad
  &&\mijCQyI(\bar{\vec{x}};k)
  &= \min_{{\vec{u}}\in\Pi_{\mpBykPY}}\max_{{\vec{d}}\in\Pi_{\mkYcAoM}}
    \mPmwBxY(\bar{\vec{x}},{\vec{u}},{\vec{d}},k)
  \notag
  \\
  \text{subject to}\qquad
  &&\bar{\vec{x}}_{k+1}
  &= \bar{\vec{x}}_k +
    \hat\mNcTCKG({\vec{u}},{\vec{d}};\bar{\vec{x}}_k,k)%
  \notag
  \\
  \text{where}\qquad
  &&\mPmwBxY(\bar{\vec{x}},{\vec{u}},{\vec{d}},k)
  &\textstyle
    = \sum_{i=k}^{N-1}
    L({\vec{u}},{\vec{d}};\bar{\vec{x}}_i,i)
    + \Phi(\bar{\vec{x}}_N,N)\;.
  \notag
\end{align}
$\mPmwBxY(\bar{\vec{x}},{\vec{u}},{\vec{d}},k)$ is the
finite-horizon, 
discrete cost-to-go comprised of the running and terminal costs of a
play from state $\bar{\vec{x}}_k$ at time $k$, using the
strategy profile $({\vec{u}},{\vec{d}})$.
$\mijCQyI(\bar{\vec{x}};k)\in[0,\top]$ with
$\top<\infty$ is the finite-horizon value corresponding to the
optimal cost-to-go, using the optimal
profile~$({\vec{u}}^*,{\vec{d}}^*)$.
Then, we define 
$\mDpbDSN(\mathcal{X},k) =
\{{\vec{x}}\in\mathcal{X}\mid
\mijCQyI({\vec{x}},k)<\top\}$ to be
an 
approximation of $G_{\vec{s}}$'s 
\emph{winning region} at stage $k$ when playing no more than
$N-k$ 
steps.  In particular,
$\mDpbDSN(\mathcal{X},N) =
\{{\vec{x}}\in\mathcal{X}\mid \Phi({\vec{x}})<\top\}$.%
\footnote{We use $\mPmwBxY[N]$, $\mijCQyI[N]$, and
  $\mDpbDSN[N]$, if the horizon $N$ is not
  clear from the context.}
To solve \eqref{l:5}, we employ a forward-Euler DDP
(\Cref{l:6}) for the period
$I=1..N$.

\begin{algorithm}
  \caption{Discrete dynamic programming} 
  \label{l:6}
\includegraphics[width=\linewidth]{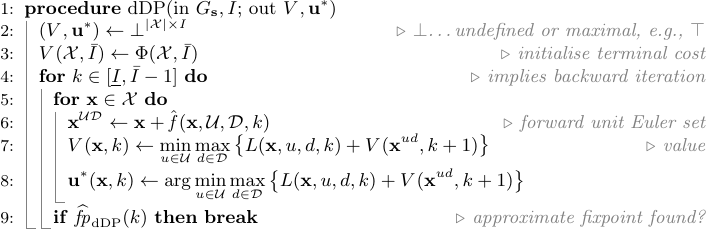}%
\end{algorithm}
 
The single-step successor
${\vec{x}}^{ud} 
= {\vec{x}} +
\hat\mNcTCKG({\vec{x}},u,d,k)%
$ 
is lifted in
Line~6, such that
${\vec{x}}^{ud}\in
{\vec{x}}^{\mpBykPY\mkYcAoM}\subseteq{\mathbb{X}}$.
Moreover, let
${\vec{x}}^{{\vec{u}}{\vec{d}}}\in\bar{\mathbb{X}}$ be the
$N$-step successor (trajectory) 
of $G$ emanating from ${\vec{x}}\in{\mathbb{X}}$ under
influence of $({\vec{u}},{\vec{d}})$, and
${\vec{x}}^{{\vec{u}}\mkYcAoM}\subseteq\bar{\mathbb{X}}$ be the family
of such trajectories 
under~$\mkYcAoM$.
For $\bar{\vec{x}}={\vec{x}}^{{\vec{u}}{\vec{d}}}\in\bar{\mathbb{X}}$
and 
$\mathcal{C}\subseteq
{\mathbb{X}}\times\mpBykPY\times\mkYcAoM\times\mathbb{N}_+$,
trajectory granularity~$\mBuHESu[]$ 
is bounded below by the smallest 
Lipschitz constant $K_*$ of $\hat\mNcTCKG$ on
$\mathcal{C}$ via the supremal norm
$\norm{\cdot}_\mathcal{C}$ 
by
\begin{align}
  \label{l:7}
  \sup_{k\in\mathbb{N}_+}
  \norm{\bar{\vec{x}}_{k+1}-\bar{\vec{x}}_k}
  \leq
  K_*
  = \norm{\hat\mNcTCKG}_\mathcal{C}
  := \sup_{c\in\mathcal{C}}\norm{\hat\mNcTCKG(c)}
  \leq \delta
  \;.
\end{align}
$\mBuHESu[]\supseteq[\pm\delta]^m$ may deliver bounds on measurement
uncertainties, for example, from feedback delay, that is, delayed observation of 
environmental effects on~${\vec{x}}$.

In Line~9, we allow a check of
$\widehat{\mathit{fp}}_{\mathrm{dDP}}$ for whether some approximation of a
fixpoint is achieved prior to reaching~$N$.  While meeting
the actual fixpoint would allow us to obtain a
stationary 
optimal strategy profile, 
the fixpoint approximation still guarantees bounded correct
(quasi-stationary and near-optimal) strategies under reduced memory
consumption and computational effort.
In \cite{Gleirscher2025-ParametricModelOptimal}, we approximate
that fixpoint with
\begin{align}
  \label{l:8}
  \widehat{\mathit{fp}}_{\mathrm{dDP}}(k)
  \equiv
  \majimMG{\mDpbDSN(\mathcal{X}_{\vec{s}},k)} 
  =
  \majimMG{\mDpbDSN(\mathcal{X}_{\vec{s}},k+1)}
  \lor 
  \widehat{\mathit{fp}}_{\vec{U}}(k)\,,
\end{align}
where the termination shortcut
$\widehat{\mathit{fp}}_{\vec{U}}(k) \equiv \exists k'\geq
k\colon
\top\not\in\mijCQyI({\vec{x}}\oplus\mBuHESu,k')$ specifies
that, latest at stage $k$, ${\vec{x}}$ is $\delta$-robustly
located inside $\mDpbDSN(\mathcal{X}_{\vec{s}},k)$.

\paragraph{Hybrid Game Playing.}

$G[\vec{U},T]$ is parameterised by a
\emph{hyper-policy} $\vec{U}$ and a \emph{task} $T$.  Our
model in \cite{Gleirscher2025-ParametricModelOptimal} assumes that
${\vec{x}}=({\vec{p}},{\vec{v}},i)$ contains a position, a velocity, and
a waypoint index.
\label{l:9}%
$\vec{U}_{\delta}\colon\mathcal{G}\times\mCBuiuW{\mathbb{N}_+}
\to\Pi_{\mpBykPY}$ increases flexibility, in particular, it
wraps \Cref{l:6} into a loop for spatio-temporal extension of
$(\mathcal{X},I)$.  $\delta$ is the robustness margin
according to \eqref{l:7}.
The role of $T = ({\mathbb{X}},{\bar{r}},\mORCysO)$, with route
${\bar{r}}\subset{\mathbb{X}}|_{\vec{p}}$ and fixed obstacle cloud
$\mORCysO\subset{\mathbb{X}}|_{\vec{p}}$, is to refine~$F$.

Importantly, $T$ needs to satisfy a healthiness condition called
\emph{$\delta$-perforation}:
\begin{align}
  \label{l:10}
  \exists \mQtakuY\in\tilde{\mathbb{X}}\colon
  {\bar{r}}\subseteq
  \underbrace{\mQtakuY|_{{\vec{p}}}\oplus\mBuHESu}_{\text{cont.\ $\delta$-tube}}
  \land
  \mQtakuY|_{{\vec{p}}}\oplus\mBuHESu 
  \cap\mORCysO=\varnothing\,.
\end{align}
Given \eqref{l:10}, \eqref{l:4}
implies $\mijCQyI(\vec 0,\cdot)=0$.
We use \Cref{l:11} as a simulator that terminates when
an overall goal $\Omega$ is reached or failures occur
(Line~10).

\begin{algorithm}
  \caption{%
    Hybrid game player (with restricted task updates)}
  \label{l:11}
\includegraphics[width=\linewidth]{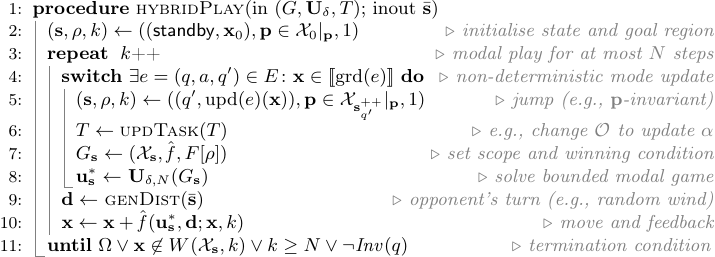}%
\end{algorithm}

\begin{definition}[Online Synthesis Problem]
  \label{l:12}
  Given a hybrid game automaton~$G$ and a task~$T$, continuously
  find tactical controllers ${\vec{u}}$ steering the system 
  along route~${\bar{r}}$ while safely circumventing unsafe
  regions. 
  For correctness, we need ${\vec{u}}$ to be
  \begin{inparaenum}[(i)]
  \item $\delta$-robust~(i.e., safe under bounded disturbance
    $\mkYcAoM$),
  \item near-optimal~(i.e., follow the minimum-cost path inside
    padded segment scopes, approximating all waypoints), and
  \item reaching the endpoint of ${\bar{r}}$.
  \end{inparaenum}
\end{definition}

\section{Side-Conditions for Game Solvability}
\label{l:13}

From our setting in \cite{Gleirscher2025-ParametricModelOptimal},
it may seem obvious that $\mijCQyI({\vec{x}},1)<\top$ is
a sufficient condition for a bounded safe and live controller
${\vec{u}}^*$ to exist.  However, it turns out to non-obvious how to
establish this condition in general and, in particular, for our
setting.  Hence, we first identify several side-conditions for game
solvability, investigate their necessity and sufficiency, and then
extract a theorem and proof from these findings.

\paragraph{Local Controllability.}

To ensure the existence of control inputs to 
reject disturbance,
a dynamics~$\hat\mNcTCKG$ 
with a balanced equivalent-order player superposition
requires the player alphabets $\mpBykPY$ and $\mkYcAoM$ to satisfy
\begin{align}
  \label{l:14}
  \forall{\vec{d}}\in\mkYcAoM\, 
  \exists{\vec{u}}\in\mpBykPY
  \colon
  \langle
  {\vec{d}}
  + \frac{\langle{\vec{d}},{\vec{u}}\rangle}{\norm{{\vec{d}}}} 
  ,
  {\vec{d}}
  \rangle\leq 0
  \,.
\end{align}
Adding ${\vec{d}}$ to the ${\vec{d}}$-component
$\frac{\langle{\vec{d}},{\vec{u}}\rangle}{\norm{{\vec{d}}}}$ of ${\vec{u}}$
in the scalar product with ${\vec{d}}$ stresses the relevance of
${\vec{u}}$'s magnitude.  Refining \eqref{l:14} from
compensation ($\leq$) 
to overriding ($<$), 
${\vec{v}}$ and, hence, ${\vec{p}}$ become locally
controllable.\footnote{Comparing $\mpBykPY,\mkYcAoM$ by diameter
  only works if $\mpBykPY,\mkYcAoM$ are centered, spherical, and
  dense 
  enough.}

\Cref{l:6} may deviate from \eqref{l:14}.  However,
if, from ${\vec{x}}$, a worst-case ${\vec{d}}$ would lead to
$\mcrHMQW\alpha$ then, by \eqref{l:14}, \Cref{l:6} can
prevent that from happening by 
prematurely rejecting ${\vec{d}}$, no matter whether or not the
environment actually chooses such a ${\vec{d}}$ when in ${\vec{x}}$.

Local controllability is neither necessary (simple tasks are still
solvable, e.g., with a short enough path to $\mcrHMQW\rho$) nor
sufficient (no guaranteed progress, e.g., if the obstacle cloud is not
enough perforated) for game solvability.

\paragraph{Lower Bound for Robustness against Disturbance.}

It is important to know how robust the system is against disturbance
under delayed observation and (minimal) non-zero progress. Equivalently, for
evaluating $\delta$-perforation \eqref{l:10}, we have
to know how much space the controller needs at least to react under
delayed 
feedback.
More generally, because of possible latencies in sensing, computation,
and actuation, disturbance rejection while maintaining correctness
(e.g., safety, no stalling) might be delayed by one or more time steps.

While \eqref{l:7} enforces cooperation
towards a maximum change of~${\vec{x}}$, players in
$G_{\vec{s}}$ act 
adversarially.  The $\delta$ in \eqref{l:10} can thus
be smaller than in \eqref{l:7}:
\begin{align}
  \label{l:15} 
  0&<
     \max_{d\in\mkYcAoM}
     \inf_{\substack{
     {\vec{v}}+\dot{\vec{v}}\neq\vec 0
     \\ 
     {\vec{x}}\in{\mathbb{X}},\, k\in\mathbb{N}_+}}
     \norm{\hat\mNcTCKG({\vec{x}},
     \vec 0, 
     d,k)%
     |_{\vec{p}}}
     <\delta\;,
  \\
  \label{l:16}
  0&<
     \max_{d\in\mkYcAoM}
     \sup_{\substack{
     {\vec{x}}\models\rho
     \\ 
     {\vec{x}}\in{\mathbb{X}},\, k\in\mathbb{N}_+}}
     \norm{\hat\mNcTCKG({\vec{x}},
     \vec 0, 
     d,k)%
     |_{\vec{p}}}
     <\delta_\rho\;.
\end{align}
Being in the center of a narrow cavity outside
$\mcrHMQW\alpha$ (e.g., in a tube), the controller does not know about
the environment's current and next choices but may want to provide
choices to make progress. 

Because of value minimisation in \Cref{l:6}, the worst case of
$u$ and
$d$ being coincident improves to the controller idling
($u=\vec
0$, e.g., because of a temporary fault) and the environment gaining 
instants of exclusive control.
In asymmetric cases (e.g., close to a wall), \Cref{l:6} will choose
to deviate from $u=\vec
0$ in either direction and perhaps reduce
$\delta$ locally even further.

The infimum in \eqref{l:15} strives for a
minimum non-zero speed $v_{\min}$ 
under disturbance rejection delayed by one stage.  
To make progress safely for at least one stage, $\delta$ then
has to be greater than $v_{\min}$.  However,
\eqref{l:15} is neither necessary (e.g.,
the system may slip through narrower cavities) nor sufficient (e.g., it
does not address global progress or safety) for game solvability.

The lower bound in \eqref{l:16}
indicates a safe region under a potential lack of control beyond
$N$ and before or after a
jump 
(e.g., at highest speed).

\paragraph{$\delta$-Robust Tracking.}

$\mpBykPY$ and $\mkYcAoM$ according to \eqref{l:14}
and $\delta$ according to
\eqref{l:15} together imply that, among
the ${\mathbb{X}}$-continuous trajectories satisfying
\eqref{l:10}, at least one $\mQtakuY\in\tilde{\mathbb{X}}$ is
\emph{robustly trackable}, that is,
\begin{align}
  \label{l:17}
  \max_{d\in\mkYcAoM}%
  \inf_{\substack{{\vec{x}}\in\mQtakuY
  \\ 
  k\in 1..\majimMG{\mQtakuY}}}
  \norm{
  \hat\mNcTCKG({\vec{x}},
     \vec 0, 
  d,k)%
  |_{\vec{p}}
  }%
  <\delta\;.
\end{align}
We observe that the nearest point on $\mQtakuY$ is within a
$\delta$-ball around
$\mQtakuY_k^{ud}$.  Equivalently, a
(controlled) move from $\mQtakuY_k$ does not deviate from
$\mQtakuY$ by more than~$\delta$.

In passages of width above $\delta$, there will be several
trajectories $\mQtakuY\in\tilde{\mathbb{X}}$ for \Cref{l:6} to construct
a controller ${\vec{u}}^*$ that can track a matching (discontinuous)
$\mBuHESu[]$-trajectory $\mQtakuY'\subseteq\mQtakuY$ (i.e., distance below
$\delta$ by definition of ${\mathbb{X}}$-continuity).  In passages
of width $\delta$, \Cref{l:6} will construct a controller
${\vec{u}}^*$ that satisfies \eqref{l:15}
and, thus,~\eqref{l:17} and, by
\eqref{l:14}, will be able to again approach the
near-optimal trajectory $\mQtakuY$ from \eqref{l:10} latest
in the subsequent stage.

If the smallest perforation width $\delta_{\mORCysO}$ found in
$\mORCysO$ is a multiple of $\delta$ according to
\eqref{l:15}, we obtain an upper
\emph{delay bound}~$n$ defined as
\begin{align}
  \label{l:18}
  n=\max\{n'\in\mathbb{N}\mid n'\delta\leq\delta_{\mORCysO}\}\;.
\end{align}
Disturbance rejection and robust tracking then extend to
$n$ 
observational delays and actuator perturbations.
We call a system \emph{delay-bounded}, if it operates with a maximum
of $n$ delays and perturbations.  Delay-boundedness, emerging from
local controllability \eqref{l:14} and robust
trackability
\eqref{l:17}, 
forms a necessary condition for the existence of a corresponding
controller. 

\paragraph{Heuristics for an Initial Horizon.}

A lower bound for $N$---under-approximating
$K_*$ and rarely applicable---would be
$\inf_{{\vec{p}}'\in\mcrHMQW\rho|_{\vec{p}}}\norm{{\vec{p}} -
  {\vec{p}}'}/v_{\max}\leq N$. 
If $\norm{{\vec{v}}}\geq v_{\min}>0$ is enabled by $\delta$
satisfying \eqref{l:15}, an upper, albeit
impractical, bound would be $N\leq\majimMG{{\mathbb{X}}|_{\vec{p}}}/v_{\min}$.
Being at ${\vec{p}}$, initialising $\vec{U}$ 
with
\begin{align}
  \label{l:19}
  N\approx
  \inf_{{\vec{p}}'\in\mcrHMQW\rho|_{\vec{p}}}
  \left\lfloor
  \frac{2\sigma\norm{{\vec{p}} - {\vec{p}}'}}{v_{\min}+v_{\max}}
  \right\rfloor
\end{align}
can be a reasonable heuristic, with sparsity~$\sigma=1$ for a sparse
or highly perforated $\mORCysO$ and $\sigma\gg 1$ for a densely populated
$\mORCysO$.  While the lower bound is a necessary condition (i.e., the
goal region must be closer than what the system can reach in
$N$ stages at highest speed), none of the bounds are
sufficient (e.g., a non-perforated $\mORCysO$ with $\sigma=\infty$ will
obstruct reaching the goal).

\paragraph{Upper Bound of Stage Cost.}

$L$ can be both zero outside
$\mcrHMQW\rho\setminus\mcrHMQW\alpha$ and non-zero inside that area.
However, $L$ must be bounded by
$N\sup_{c\in\mathcal{C}}L(c)<\top$, which 
is a necessary condition for the following discussion.

\section{Solvability and Controlled Invariant Set} 
\label{l:20}

For the discussion below, we fix $G_{\vec{s}}$ for some
${\vec{s}}\in{\mathcal{S}}$ and leave ${\vec{s}}$ implicit.  Let
${\vec{x}}\in{\mathbb{X}}$,
$(u, d)\in\mpBykPY\times\mkYcAoM$,
$k\in 2..%
N$.
For the inclined reader, it is easy to see that
Line~8
of \Cref{l:6} implement a 
discrete HJI equation~\cite{Tomlin2000-gametheoreticapproach}
\begin{align}
  \label{l:21}
  \mijCQyI({\vec{x}},k-1) - \mijCQyI({\vec{x}},k)
  =
  \min_{u\in\mpBykPY}
  \max_{d\in\mkYcAoM}
  \mathcal{H}({\vec{x}},u,d,k)
\end{align}
with a corresponding 
Hamiltonian (i.e., directional difference of 
$\mijCQyI$ along 
$\hat\mNcTCKG$) 
\begin{align*}
  \mathcal{H}({\vec{x}},u,d,k)
    &=
    L({\vec{x}},u,d,k)
    + 
    \mijCQyI({\vec{x}}^{ud},k)
    -
    \mijCQyI({\vec{x}},k)
    \;.
\end{align*}
HJI equations connect value evolution with the minimax
combination of stage cost and 
value descent. 
Below, we first establish a fundamental lemma and proceed with an
inductive invariant.

\begin{lemma}[Monotonicity] 
  \label{l:22}
  $\mijCQyI({\vec{x}},k-1)<\top
  \implies
  \mijCQyI({\vec{x}},k-1)
  \leq\mijCQyI({\vec{x}},k)$
\end{lemma}

\begin{proof}[\Cref{l:22}]
  By Bellman's principle of optimality,
  $\mijCQyI({\vec{x}},k-1)$ is the value of the optimal
  suffix of 
  length $N-(k-1)$.  Exhaustive 
  search (Lines~4 to~6 in
  \Cref{l:6}) rules out shorter suffixes of length
  $N-k$, yet unseen but better with respect to $\mPmwBxY$,
  such that
  $\mijCQyI({\vec{x}},k-1)>\mijCQyI({\vec{x}},k)$.
\end{proof}

\begin{lemma}[Reach-Avoid $k$-Winning Region]
  \label{l:23} 
  If $\mijCQyI$ (obtained via
  $\vec{U}$) 
  satisfies \eqref{l:21} then
  $\mDpbDSN(\mathcal{X},k)$ contains exactly the states from
  where the system can reach $\mcrHMQW\rho\setminus\mcrHMQW\alpha$ in at
  most $N-k$ stages while remaining outside
  $\mcrHMQW\alpha$.
\end{lemma}

\begin{proof}[Sketch for \Cref{l:23}]
  The proof is by reverse induction on $k$ in
  \eqref{l:21} and provided in full detail in
  \Cref{l:29}.  This proof provides that
  \Cref{l:23} specifies an inductive
  invariant.
\end{proof}

Under well-formedness assumptions detailed below,
\Cref{l:23} is a sufficient condition for
$k$-bounded safe and live controllers to exist.

\paragraph{Configurations and Pre-/Post-Conditions.}

Defined in \Cref{l:2}, the hybrid game
automaton~$G$, the hyper-policy $\vec{U}_{\delta}$
with a robustness margin~$\delta$, and the task
specification~$T=({\mathbb{X}},{\bar{r}},\mORCysO)$ together form
a \emph{configuration}
$\mathsf{C}=(G, \vec{U}_{\delta},
T)$.
We call $\mathsf{C}$ \emph{well-formed}, written
$\mathit{wf}(\mathsf{C})$, if
\begin{enumerate}[(i)]
\item $\mORCysO$ is $\delta$-perforated \eqref{l:10}
  with $\delta$ according to
  \eqref{l:15},
\item locally controllable \eqref{l:14}, and
\item delay-bounded \eqref{l:18}.
\end{enumerate}
For $G_{\vec{s}}$ and a specific initial horizon $N$
according to \eqref{l:19}, we obtain the \emph{modal}
configuration
$\mathsf{C}_{\vec{s}}=(G_{\vec{s}},
\vec{U}_{\delta,N}, T)$ with discrete
dynamics~$\hat\mNcTCKG$. 
We have $\mathit{wf}(\mathsf{C}_{\vec{s}})$ if $\mathit{wf}(\mathsf{C})$
and $G_{\vec{s}}$ is initialised with
${{\vec{x}}_0}\not\in\mcrHMQW{\alpha}$.

We call ${\vec{u}}^*_{\vec{s}}$ \emph{correct}, written
$\mathit{cr}({\vec{u}}^*_{\vec{s}})$, if it is
\begin{enumerate}[(i)]
\item live (i.e., it reaches the next waypoint or finishes its task),
\item safe (i.e., it avoids the unsafe region, e.g., it avoids collisions),
\item $\delta$-robust (it rejects disturbances
  below $\delta$ with a single-step delay without getting into an
  unsafe state), and
\item near-optimal (it provides the steepest value-decreasing control input),
\end{enumerate}
when applied to ${\vec{x}}\in\mDpbDSN(\mathcal{X}_{\vec{s}},1)$
for a maximum of $N$ stages.

The key aspects of the proposed online synthesis are summarised in the
following solvability theorem.

\begin{theorem}[Modal Game Solvability]
  \label{l:24}
  Given a modal configuration $\mathsf{C}_{\vec{s}}$, if
  $\mathsf{C}_{\vec{s}}$ is \emph{well-formed} then
  $\vec{U}$ 
  constructs 
  a \emph{correct} controller ${\vec{u}}^*_{\vec{s}}$,
  \begin{align*}
    \forall{\vec{s}}\in{\mathcal{S}}\colon
    \{\mathit{wf}(\mathsf{C}_{\vec{s}})\}\;
    {\vec{u}}^*_{\vec{s}} :=
    \vec{U}_{\delta,N}(G_{\vec{s}})\;
    \{\mathit{cr}({\vec{u}}^*_{\vec{s}})\}\;.
  \end{align*}
\end{theorem}

\begin{proof}[Sketch for \Cref{l:24}]
  The post-conditions (i) to (iv) are established separately based on
  \eqref{l:10}, the side conditions in
  \Cref{l:13}, features of \Cref{l:6},
  and \Cref{l:23}.  The complete proof is
  provided in \Cref{l:29}.
\end{proof}

\begin{lemma}[Preservation]
  \label{l:25}
  $\vec{U}_{\delta,N}$
  preserves game solvability.
\end{lemma}

\begin{proof}[\Cref{l:25}]
  Observe that ${\vec{x}}$ and $\rho$ are fixed during scope
  extension.  By \Cref{l:22}, if there exists a correct
  solution within $N$ stages, a corresponding solution exists
  within $N+1$ stages.  That solution will just not allow the
  controller to wait for too long in order to remain on the safe side.

  Under local controllability \eqref{l:14}, extending the
  prediction 
  interval $I$ (initially
  $1..N$) 
  to $I'=\underline I..\bar I+{\delta_I}$
  and $\mathcal{X}_{\vec{s}}$ to some
  $\mathcal{X}'\supset\mathcal{X}_{\vec{s}}$ preserves
  \begin{inparaenum}[(i)]
  \item liveness by
    relaxing~$\mDpbDSN[\bar I'](\mathcal{X}',1)$
    (from ${\vec{x}}$, the controller has more space and time to steer
    the system into $\mcrHMQW\rho\setminus\mcrHMQW\alpha$),
  \item safety by
    constraining~$\mDpbDSN[\bar I'](\mathcal{X}',1)$
    only as far as necessary (allowing at most $n$ delays
    \eqref{l:18} with $\delta'=n\delta$, the
    environment has only $n$ stages to accumulate unsafe
    disturbances)
  \end{inparaenum}
  and, thus, solvability.
\end{proof}

\begin{corollary}[Controlled Invariant Set] 
  With
  $\mDpbDSN(\mathcal{X}_{\vec{s}},1) \subseteq
  {\mathbb{X}}\setminus(\mcrHMQW\alpha\oplus\mBuHESu)$,
  $\vec{U}_{\delta,N}$ 
  provides a maximum
  $N$-bounded, robust (i.e., under-approximated),
  non-linear controlled invariant set, independent of
  \eqref{l:14}.
\end{corollary}

\Cref{l:24} establishes sub-problems (i) and (ii) of
\Cref{l:12}.  To finally address sub-problem
(iii), we define
$\mathit{cr}(\bar{\vec{s}})\equiv\bar{\vec{s}}_n|_{\vec{p}}={\vec{p}}_n$ and provide a
solvability theorem applicable for the considered range of
modal games.

\begin{theorem}[Global Game Solvability]
  \label{l:26}
  Given a configuration $\mathsf{C}$, if $\mathsf{C}$ is
  \emph{well-formed} then \Cref{l:11} produces a
  \emph{correct} play $\bar{\vec{s}}$ with
  $\bar{\vec{s}}_0={\vec{s}}_0$, 
  \begin{align*}
    \{\mathit{wf}(\mathsf{C})\}\;
    \bar{\vec{s}}:=\textsc{hybridPlay}[\mathsf{C}](\bar{\vec{s}})\;
    \{\mathit{cr}(\bar{\vec{s}})\}\;.
  \end{align*}
\end{theorem}

\begin{proof}[Sketch \Cref{l:26}]
  The proof demonstrates the existence and correctness of a fixpoint
  $\bar{\vec{s}}$ and is provided in full detail in
  \Cref{l:29}.
\end{proof}

\section{Discussion}
\label{l:27}

\paragraph{From Non-Stationary to Quasi-Stationary Controllers.}

If we reach the 
actual value fixpoint $\mijCQyI[\infty]$ in stage
$k^\infty\in 1..N$, we obtain a stationary,
infinite-horizon~${\vec{u}}^\infty$, which provides the
$\mathcal{X}$-optimal input independent of time, forming a winning
strategy as long as it has \emph{at least $N-k^\infty$
  stages} to solve the problem.

Otherwise, we obtain a non-stationary, $(\mathcal{X},k)$-optimal,
terminal-time controller.  This is one, which, at each stage
$k$, provides an $(\mathcal{X},k)$-optimal input, forming a
winning strategy, given \emph{exactly $N-k$ stages} to
solve the problem. 

However, applying a fixpoint approximation in \Cref{l:6}, such as
\eqref{l:8}, and reaching the quasi-fixpoint
$\mijCQyI[*]$ in stage $k^*$, we obtain a 
quasi-stationary~${\vec{u}}^*$.  This is one, which provides a
\emph{near-$\mathcal{X}$-optimal} input, forming a winning strategy as
long as it requires \emph{at most $N-k^*$ stages} to
solve the problem.

By \Cref{l:22}, we can use ${\vec{u}}({\vec{x}},k)$
across all the stages in $k..N$\footnote{That is, it
  can be used for a maximum of $N-k$ stages.}
while preserving correctness and benefitting from lowering the space
complexity $\mmHxeoG{N\mathcal{X}}$
of the computation and storage of ${\vec{u}}$ to
$\mmHxeoG{\mathcal{X}}$. 
Regarding optimality, for $k<k'$,
${\vec{u}}({\vec{x}},k)$ performs at least as good as
${\vec{u}}({\vec{x}},k')$
\cite[Fig.~9]{Gleirscher2025-ParametricModelOptimal}.
Moreover, using ${\vec{u}}({\vec{x}},k^*)$ allows one to omit the
computation of the stages $1..k^*$.
However, one reason to prefer the non-stationary over the
(quasi-)stationary controller could be situations where
${\vec{u}}$ is computed in the presence of time-varying dynamics.

\paragraph{Requirements on Shielding imposed by the Discretisation.}

Unsurprisingly, our choice of DDP simplifies shielding.  In
particular, for pre-shielding in
Line~8 of
\Cref{l:6}, we use a single-step check against
entering~$\mcrHMQW\alpha$.
More sophisticated shields can 
be applied \cite{Alshiekh2018-SafeReinforcementLearning}.
In certain 
reinforcement learning~(RL) schemes, temporal property shielding (e.g., safe or co-safe
LTL) can be used for achieving correctness.  However, \Cref{l:6}
reduces this issue to single-step shielding (e.g., ensure
$\mAoguFS{X}\neg\alpha$ to obtain $\mAoguFS{G}\neg\alpha$) against
unsafe jumps due to the discretisation, 
of course under the assumption that the winning condition~$F$
adequately captures safety and performance requirements.

\paragraph{Generalisation and Scalability.}

Our approach is related to safe variants of, usually faster, online-RL
techniques: The structure of $F$
\eqref{l:4} corresponds to a weighted barrier
function. Hence, safety and liveness are preserved under
value function approximation as long as the replacement of
\Cref{l:6} (e.g., by model-free deep-RL) keeps
using a temporal safety pre-shield and $\delta$-perforation
\eqref{l:10} holds with $\delta$ according
to~\eqref{l:15}.  Any significant gain in
scale while supporting general settings comes with the issues
summarised in
\cite[Sec.~1]{Gleirscher2025-ParametricModelOptimal}. 

\paragraph{Relating
  \Cref{l:24,l:26} to Aerial
  Delivery \cite{Gleirscher2025-ParametricModelOptimal}.}

Waypoint areas ${\vec{p}}_i\oplus\mBuHESu$ (incl.\ ${{\vec{x}}_0}$ and
${\vec{p}}_n$) are required to be free of static obstacles.
For flexible transitions between route segments, we initialise
$\Phi$ (Line~3 of \Cref{l:6})
with a cuboid around the next segment,
including~${\vec{p}}_i\oplus\mBuHESu$.

Regarding robustness \eqref{l:15}, the aerial
delivery setting terminates in
$(\mcrHMQW\rho\setminus\mcrHMQW\alpha)\oplus\mBuHESu$ because of the
single-delay scheme with $\mBuHESu=[\pm\delta]^3=[\pm 2]^3$.
This is safe as we define $\alpha$ such that 
$\mcrHMQW\rho\setminus\mcrHMQW\alpha$ is outside
$\mORCysO\oplus\mBuHESu$.

Control towards $v_{\min}$ (i.e., braking) to permit fine
$\delta$-perforation (i.e., small $\delta$) can increase the
minimum $N$ required for liveness.  Determining a
maximum
$N$ may need scenario-specific experimentation.

Regarding the upper bound for step-width ($\mBuHESu[]$) in
\Cref{l:6}, it can be seen from
\eqref{l:7} that a hypothetical
per-component $\delta$ might be much smaller than $K_*$.
Particularly, given a restriction of
$\hat\mNcTCKG$ 
to maintain $\norm{{\vec{v}}}\leq v_{\max}$, we can reduce
$\mathcal{C}$ in \eqref{l:7} to the
maximum absolute changes per time unit and assume
$\mBuHESu[]=[\pm v_{\max}]^3\times[\pm 2]^3\times\{0,1\}$ for the
coarsest trajectories.  For that, physical limits (e.g., of~${\vec{v}}$)
are assumed to be over-approximated (e.g., $v_{\max}$) during
$\mathcal{X}$-constrained exploration in Line~6.

\paragraph{Potential Relevance for Certification.}

For example, our argument for controller robustness and our
formalisation of online synthesis respond to requirement~2110
(providing a robust flight envelope) as well as requirement~2265
(quality assurance) of EASA
SC-VTOL-01~\cite{EASA2019-SpecialConditionsmall}, an airworthiness
standard applicable for type certification of person-carrying small
aerial vehicles.

\section{Conclusion}
\label{l:28}

This work extends our work in
\cite{Gleirscher2025-ParametricModelOptimal} by a solvability (i.e.,
controller existence) theorem for approximate (i.e., stage-bounded,
discretised) modal reach-avoid games.  This theorem uses a numerical
HJI equation as a generic inductive invariant of the
incrementally computed winning region. 
Necessary and sufficient conditions for solvability are
identified.
This invariant forms an approximate control barrier certificate,
such that the synthesised 
controllers guarantee $\delta$-robust bounded correctness, in
particular, safety for at least $N$ stages and liveness by
reaching goals within $N$ stages; both under rejection of
delayed bounded disturbance.
Under the identified conditions, we obtain a decision
procedure 
for approximate 
modal reach-avoid games using a DDP-based
$\mmHxeoG{N\majimMG{{\mathbb{X}}}\majimMG{\mpBykPY}\majimMG{\mkYcAoM}}$-algorithm.

In future work, we will enhance our reach-avoid reasoning into a
richer deductive framework, feeding into an assurance workflow for
cyber-physical system. 
One of the obvious performance questions will be to investigate the
most beneficial ways in which variants of 
approximate dynamic programming, model-predictive control, and deep-RL can preserve
\Cref{l:24}.

\bibliographystyle{plain}
\bibliography{main}

\appendix
  \appendix

\section{Proofs}
\label{l:29}
  
\begin{proof}[\Cref{l:23}]
  The proof of the reach-avoid-invariant winning region is by reverse
  induction on $k$.

  \emph{Induction start} ($k=N$):
  From \eqref{l:5} and \eqref{l:4}, we
  obtain $\mijCQyI({\vec{x}},N) =\Phi({\vec{x}})$ (cf.
  Line~3 of \Cref{l:6}) and
  $\{{\vec{x}}\in\mathcal{X}\mid
  \mijCQyI({\vec{x}},N)<\top\} =
  \{{\vec{x}}\in\mathcal{X}\mid
  \Phi({\vec{x}})<\top\}$ contains exactly the states in
  $\mcrHMQW\rho\setminus\mcrHMQW\alpha$, those reachable in 0 stages,
  hence, $\mDpbDSN(\mathcal{X},N)$.
  Given
  $L(\cdot,N)=0$,
  \eqref{l:21} also satisfies the
  induction start:
  \begin{align*}
    \mijCQyI({\vec{x}},N-1) -
    \mijCQyI({\vec{x}},N)
    &=
    \min_{u\in\mpBykPY}
    \max_{d\in\mkYcAoM}
    \big\{
    L({\vec{x}},u,d,N)%
    +
    \mijCQyI({\vec{x}}^{ud},N) -
    \mijCQyI({\vec{x}},N)
    \big\}
    \\\iff\qquad
    \mijCQyI({\vec{x}},N-1)
    &=
      \min_{u\in\mpBykPY}
      \max_{d\in\mkYcAoM}
      \Phi({\vec{x}}^{ud})
  \end{align*}
  If
  ${\vec{x}}\in\mDpbDSN(\mathcal{X},N)$
  then by \Cref{l:22}, we have
  $\mijCQyI({\vec{x}},N-1)<\top$ and, thus, also
  ${\vec{x}}^{ud}\in
  \mcrHMQW\rho\setminus\mcrHMQW\alpha$. 
  \hfill\resizebox{.5em}{.6em}{/\kern-.5em$\Box$}

  \emph{Induction step} ($k\to k-1$): Given
  $\mDpbDSN(\mathcal{X},k)$ according to
  \Cref{l:23} (induction hypothesis, IH),
  show that
  $\mDpbDSN(\mathcal{X},k-1)$
  satisfies \Cref{l:23}.

  We first show that
  $\{{\vec{x}}\in\mathcal{X}\mid
  \mijCQyI({\vec{x}},k-1)<\top\} \subseteq
  \mDpbDSN(\mathcal{X},k-1)$.
  Fixing ${\vec{x}}\in\mathcal{X}$, we use both sides of
  \eqref{l:21}:
  \begin{itemize}
  \item By IH, $\mijCQyI({\vec{x}},k)<\top$.  From
    \Cref{l:22} and the left-hand side~(LHS) of
    \eqref{l:21}, we then obtain
    $-\top<\mijCQyI({\vec{x}},k-1) -
    \mijCQyI({\vec{x}},k)\leq 0$.
  \item 
    From the right-hand side~(RHS) of \eqref{l:21}, we now require
    \begin{align*}
      &\min_{u\in\mpBykPY}
      \max_{d\in\mkYcAoM}
      \big\{
      L({\vec{x}},u,d,k)%
      +
      \mijCQyI({\vec{x}}^{ud},k) -
      \mijCQyI({\vec{x}},k)
      \big\}
      \leq 0
      \\\iff&
      \min_{u\in\mpBykPY}
      \max_{d\in\mkYcAoM}
      \big\{
      L({\vec{x}},u,d,k)%
      +
      \mijCQyI({\vec{x}}^{ud},k)
      \big\}
      \leq\mijCQyI({\vec{x}},k)\;.
    \end{align*}
    By IH, both summands must be $<\top$.  By
    \eqref{l:4},
    \begin{enumerate}[(i)]
    \item ${\vec{x}}$ must be outside $\mcrHMQW\alpha$ (ensured by safety
      pre-shielding),
    \item[(ii)] ${\vec{x}}$ is either in the goal region
      $\mcrHMQW\rho\setminus\mcrHMQW\alpha$ or given a non-$\top$
      weight, and
    \item[(iii)] $\mijCQyI({\vec{x}}^{ud},k)<\top$,
      which, by IH, means a move to
      ${\vec{x}}^{ud}\in\mDpbDSN(\mathcal{X},k)$.
    \end{enumerate}
  \end{itemize}
  Consequently,
  $\mDpbDSN(\mathcal{X},k-1)$
  \emph{contains} the states from where the system can reach
  $\mcrHMQW\rho\setminus\mcrHMQW\alpha$ in at most
  $N-(k-1)$ stages while remaining
  outside~$\mcrHMQW\alpha$.
  \hfill\resizebox{.5em}{.6em}{/\kern-.5em$\Box$}

  We then show that
  $\mDpbDSN(\mathcal{X},k-1)
  \subseteq
  \{{\vec{x}}\in\mathcal{X}\mid
  \mijCQyI({\vec{x}},k-1)<\top\}$.
  Proceeding by contradiction, we state that there is an
  ${\vec{x}}\in \mDpbDSN(\mathcal{X},k-1)$ with
  $\mijCQyI({\vec{x}},k-1)=\top$.  From that
  ${\vec{x}}$, the system either cannot reach
  $\mcrHMQW\rho\setminus\mcrHMQW\alpha$ in at most
  $N-(k-1)$ stages or it
  enters~$\mcrHMQW\alpha$ before that.
  Then, by \Cref{l:22}, we have
  $\mijCQyI({\vec{x}},k)=\top$. 
  With the LHS of \eqref{l:21} becoming
  $\top-\top=0$, we now require the RHS of
  \eqref{l:21} to satisfy
  \begin{align*}
    \min_{u\in\mpBykPY}
    \max_{d\in\mkYcAoM}
    \big\{
    L({\vec{x}},u,d,k)%
    +
    \mijCQyI({\vec{x}}^{ud},k)
    \big\}
    & \leq\top\;.
  \end{align*}
  But then at least one summand must be brought to $0$ or both to
  $<\top$.  Consequently, either
  \begin{inparaenum}[(i)]
  \item the move to
    ${\vec{x}}^{ud}$ must again result in
    $\mijCQyI({\vec{x}}^{ud},k)<\top$
    (contradiction),
  \item[(ii-a)]
    $L({\vec{x}},u,d,k)%
    <\top$, meaning $\mcrHMQW\rho\setminus\mcrHMQW\alpha$ has
    been reached already (contradition), \emph{or}
  \item[(ii-b)] a \emph{safe move} (with non-$\top$ weight) to
    some ${\vec{x}}^{ud}$.  However, because
    of the coupling of (ii-b) and (i) via $({\vec{u}}^*,{\vec{d}}^*)$,
    ${\vec{x}}^{ud}$ must be outside
    $\mDpbDSN(\mathcal{X},k)$
    (otherwise, IH would apply; contradiction).  But that would
    again contradict \eqref{l:21}.
  \end{inparaenum}
  
  This closes the proof and \Cref{l:23}
  specifies an inductive invariant.
\end{proof}

\begin{lemma}
  \label{l:30}
  Given $k\geq l\geq1$, we have that
  $\widehat{\mathit{fp}}_{\vec{U}}(k) \implies
  \widehat{\mathit{fp}}_{\vec{U}}(l)$.
\end{lemma}

\begin{proof}[\Cref{l:30}]
  \begin{align*}
    \widehat{\mathit{fp}}_{\vec{U}}(k)
    \equiv\;
    & \exists k'\geq k\colon
      \top\not\in\mijCQyI({\vec{x}}\oplus\mBuHESu,k')
      \tag{by def}
    \\\implies
    & \exists k'\geq l\colon
      \top\not\in\mijCQyI({\vec{x}}\oplus\mBuHESu,k')
      \tag{by assumption}
    \\\equiv\;
    & \widehat{\mathit{fp}}_{\vec{U}}(l)
  \end{align*}
\end{proof}

\begin{proof}[\Cref{l:24}]
  We separately establish each of the four post-conditions (i) to (iv):
  \begin{enumerate}[(i)]
  \item Assuming well-formedness $\mathit{wf}(\mathsf{C}_{\vec{s}})$, via
    \Cref{l:23}, ${\vec{u}}^*_{\vec{s}}$ guarantees
    step-wise \emph{progress} for $k\in1..N$ and,
    thus, \emph{liveness} for
    $k=1$. 
    \hfill\resizebox{.5em}{.6em}{/\kern-.5em$\Box$}
    
  \item Via \Cref{l:23}, ${\vec{u}}^*_{\vec{s}}$
    guarantees \emph{safety} for $k\in1..N$.  At each
    stage $k$, Line~6 of \Cref{l:6}
    explores all possible moves.  The computation of $L$ in
    Line~8 includes predictive run-time
    enforcement for excluding incorrect actions, in particular,
    \begin{inparaenum}[(a)]
    \item implicit pre-shielding from visiting unsafe states and
    \item explicit pre-shielding from crossing\footnote{The split into
        implicit and explicit pre-shielding merely results from
        numerical approximation.  Moreover, the supervisor not
        discussed here represents a post-shield.} unsafe states.
    \end{inparaenum}
    \hfill\resizebox{.5em}{.6em}{/\kern-.5em$\Box$}

  \item ${\vec{u}}^*_{\vec{s}}$ guarantees \emph{$\delta$-robust tracking}
    under $n$-stage delayed 
    feedback,
    given $\delta$ satisfies \eqref{l:14} and
    \eqref{l:17} and $n$ according to
    \eqref{l:18}.  From Line~7
    and \eqref{l:4}, we observe that
    $\mijCQyI({\vec{x}},k) < \top$ for states ${\vec{x}}$
    with a distance to $\alpha$ of at least~$\delta$.
    We also obtain the \emph{vector-field certificate}
    \[
      \forall k\in1..N,\,{\vec{x}}\in
      \underbrace{
        \mDpbDSN(\mathcal{X}_{\vec{s}},k)
      }_{\parbox{2cm}{\scriptsize\centering
          $k$-winning\\region}} 
      \colon
      \underbrace{
        {\vec{x}}^{{{\vec{u}}^*_{\vec{s}}({\vec{x}},k)}\mkYcAoM}|_{\vec{p}}
      }_{\parbox{2cm}{\scriptsize\centering single-step\\successor}} 
      \cap
      \underbrace{
        \mcrHMQW\alpha%
      }_{\parbox{1cm}{\scriptsize\centering
          unsafe\\region}}%
      =\varnothing\;,
    \]
    which can be checked point-wise where desired, for example, at the
    boundary
    $\partial\mDpbDSN(\mathcal{X}_{\vec{s}},N)$
    of the $N$-winning region.
    Note that the induction start of
    \Cref{l:23} only allows for a
    single-step delay.
    \hfill\resizebox{.5em}{.6em}{/\kern-.5em$\Box$}
  
  \item \emph{Near-optimality} of ${\vec{u}}^*_{\vec{s}}$ is
    established in Line~8 of \Cref{l:6} by
    choosing control inputs point-wise minimising $\lambda$.
    $\vec{U}$ 
    fails, however, to achieve global optimality for
    ${\vec{u}}^*_{\vec{s}}$ because the search for the minimum-weight
    (not shortest) path ends with achieving
    $\mijCQyI({\vec{x}},1)<\top$.  \hfill\resizebox{.5em}{.6em}{/\kern-.5em$\Box$}
  \end{enumerate}
  $\vec{U}$ 
  provides a solution to \eqref{l:21}
  for $\mijCQyI$ (and ${\vec{u}}^*_{\vec{s}}$).
  By \Cref{l:23},
  $\mDpbDSN(\mathcal{X}_{\vec{s}},1)$ is the
  (controlled) reach-avoid $N$-winning region.  Everywhere in
  $\mDpbDSN(\mathcal{X}_{\vec{s}},1)$, ${\vec{u}}^*_{\vec{s}}$ can
  be applied safely for at least $N$ steps and the value
  difference will be non-positive.
  Consequently,
  $\mDpbDSN(\mathcal{X}_{\vec{s}},1)$ is both an
  $N$-bounded (difference) invariant and variant set of
  $\hat\mNcTCKG$ under~${\vec{u}}^*_{\vec{s}}$.

  The \emph{existence} of ${\vec{u}}^*_{\vec{s}}$ equals to the
  termination of the assignment
  ${\vec{u}}^*_{\vec{s}}:=
  \vec{U}_{\delta,N}(G_{\vec{s}})$.
  According to the condition $\widehat{\mathit{fp}}_{\mathrm{dDP}}$ for
  premature termination in
  Line~9 of \Cref{l:6},
  $\vec{U}$ 
  terminates either (i) after a maximal 
  extension of $\mathcal{X}_{\vec{s}}\subseteq{\mathbb{X}}$ and
  $N$ or (ii) when $\widehat{\mathit{fp}}_{\vec{U}}(k)$
  is met.
  $\delta$-perforation \eqref{l:10} ensures (ii) to
  hold for some $k$ and the corollary
  $\widehat{\mathit{fp}}_{\vec{U}}(k)\implies
  \widehat{\mathit{fp}}_{\vec{U}}(1)$ of
  \Cref{l:30} ensures a \emph{non-empty}
  $\mDpbDSN(\mathcal{X}_{\vec{s}},1)$ including~${{\vec{x}}_0}$.
\end{proof}

\begin{proof}[\Cref{l:26}]
  The theorem states that, given $\mathit{wf}(\mathsf{C})$ and if the
  play $\bar{\vec{s}}$ corresponds to the fixpoint of
  $\bar{\vec{s}}:=\textsc{hybridPlay}[\mathsf{C}](\bar{\vec{s}})$,
  then $\bar{\vec{s}}$ is correct, $\mathit{cr}(\bar{\vec{s}})$.
  For our 
  hybrid game setting, we show that this fixpoint exists and is
  correct.

  $\bar{\vec{s}}$ exists if \Cref{l:11} terminates.
  For that, we require the loop exit condition
  (\Cref{l:6} Line~9)
  to be established.
  We can infer $\Omega$ from
  \begin{inparaenum}[(i)]
  \item the route ${\bar{r}}$ being finite,
  \item the termination of
    $\vec{U}$ 
    by \Cref{l:24} 
    after the safe approach of each waypoint~$i$, 
    and
  \item $i$ being incremented with \emph{each}
    jump in $G$ except for $\mTorQlx{finish}$.
  \end{inparaenum}
  The other lines in \Cref{l:11} are 
  assumed to be terminating.
  If \Cref{l:11} terminated early for the other three
  reasons, some of the well-formedness conditions must have been
  violated. 
  \hfill\resizebox{.5em}{.6em}{/\kern-.5em$\Box$}
    
  Reaching waypoint~$n$,
  $\bar{\vec{s}}_n|_{\vec{p}}={\vec{p}}_n$, 
  $\mTorQlx{finish}$ is enabled and establishes $\Omega$.  Hence,
  the fixpoint $\bar{\vec{s}}$ is correct.
\end{proof}

 \end{document}